\documentclass[12pt]{article}

\usepackage[cmex10]{amsmath}
\usepackage{amsfonts}
\usepackage{amssymb}
\usepackage{amsthm}

\theoremstyle{definition} \newtheorem{theorem}{Theorem}[section]
\theoremstyle{definition} \newtheorem{definition}[theorem]{Definition}
\theoremstyle{definition} 
\theoremstyle{definition} \newtheorem{proposition}[theorem]{Proposition}
\theoremstyle{definition} 
\theoremstyle{definition} 
\theoremstyle{definition} 
\theoremstyle{definition} 
\theoremstyle{definition} 
\theoremstyle{definition} 
\begin{document}

\title{Fairness in Combinatorial Auctioning Systems}
\date{}
\author{\begin{tabular}[t]{c@{\extracolsep{2em}}c}
    Megha Saini \hspace{1in} Shrisha Rao\footnote{Corresponding author.}  \\
    {\tt \{megha.saini,srao\}@iiitb.ac.in} \\
    International Institute of Information Technology - Bangalore \\
    Bangalore 560 100 \\ India
\end{tabular}}
\maketitle

\begin{abstract}
  One of the Multi-Agent Systems that is widely used by various
  government agencies, buyers and sellers in a market economy, in such
  a manner so as to attain optimized resource allocation, is the
  Combinatorial Auctioning System (CAS).  We study another important
  aspect of resource allocations in CAS, namely \emph{fairness}.  We
  present two important notions of fairness in CAS, \emph{extended
    fairness} and \emph{basic fairness}.  We give an algorithm that
  works by incorporating a metric to ensure fairness in a CAS that
  uses the Vickrey-Clark-Groves (VCG) mechanism, and uses an algorithm
  of Sandholm to achieve optimality.  Mathematical formulations are
  given to represent measures of extended fairness and basic
  fairness.
\end{abstract}


\section{Introduction}

Multi-Agent Systems (MAS) have been an interesting topic in the areas
of decision theory and game theory.  MAS are composed of a number of
autonomous agents.  In some applications, these autonomous agents act
in a self-interested manner in their dealings with numerous other
agents.  Even in game theory, in an interactive framework the decision
of one agent often affects that of another.  This behavior is seen in
the MAS which mainly deal with issues like resource
allocation~\cite{bredin00gametheoretic,sycara98}.  In such scenarios,
each agent holds different preferences over the various possible
allocations and hence, concepts like individual rationality, fairness,
optimality, efficiency, etc., are important~\cite{mara-survey}.  In
this paper, we study a framework where optimality is a desirable
property but fairness is a required property.  An excellent example of
such a framework is Combinatorial Auctioning Systems (CAS) where the
two most important issues pertaining to resource allocation are
\emph{optimality} and \emph{fairness}.

Incorporation of fairness into game theory and economics is a
significant issue.  Its welfare implications in different systems were
explored by Rabin~\cite{rabin93}.  The problem of fair allocation is
being resolved in various MAS by using different procedures depending
upon the technique of allocation of goods and the nature of goods.
Brams and Taylor give the analysis of procedures for dividing
divisible and indivisible items and resolving disputes among the
self-interested agents~\cite{brams96}.  Some of the procedures
described by them include the ``Divide and Choose'' method of
allocation of divisible goods among two agents to ensure the fair
allocation of goods which also exhibits the property of
``envy-freeness,'' a property first introduced by
Foley~\cite{foley67}.  Lucas' method of markers and Knaster's method
of sealed bids are described for MAS comprising more than two players
and for the division of indivisible items.  The Adjusted-Winner (AW)
procedure is also defined by Brams~\cite{brams05} for envy-freeness
and equitability in two-agent systems.  Various other procedures like
moving knife procedures for cake cutting are defined for the MAS
comprising three or more agents~\cite{brams05, barbane04}.

However, it can also be seen that the definition of fairness varies
across the different multi-agent systems, i.e., the term
\emph{fairness} is perceived differently in various MAS with regard to
the resource allocation.  In some MAS, it can be defined as equitable
distribution of resources such that each recipient believes that it
receives its fair share.  Thus, each agent likes its share at least as
much as that of other agents' share and, thereby, it is also known as
envy-free division of resources~\cite{brams05}.  But this definition
of fairness is not applicable to all the MAS.  To explain the notions
of fairness in MAS, we classify fairness into \emph{extended fairness}
and \emph{basic fairness} in this paper.

To illustrate these notions of fairness mathematically, we shall use
the framework of the Combinatorial Auctioning Systems (CAS).  The CAS
is a kind of MAS whereby the bidders can express preferences over
combination of items~\cite{nisan00,narahari05}.  The CAS approach is
being used by different government agencies like the
FCC~\cite{cramton05} and numerous business applications like logistics
and transportation~\cite{caplice03, caplice05}, supply chain
formation~\cite{walsh00}, B2B negotiations~\cite{jones00}, etc.  It
has been noticed that one of the significant issues in CAS is that of
resource allocation.  Optimum resource allocation is one of the most
desirable properties in a CAS, and deals mainly with the Winner
Determination Problem (WDP)~\cite{sandholm02, naramunchi05}.
Determining the winner in a CAS so as to maximize revenue is an
NP-complete problem.  However, it is seen that besides WDP, fairness
is another important objective in many CAS-like government auctions.
Rothkopf expressed his view in~\cite{rothkopf01} that ``optimal
solution to the winner determination problem, while desirable, is not
required.  What is required is a guarantee that the auction will be
fair and will be perceived as fair.''  Hence, we realize the
significance of fairness in CAS.

We shall consider a CAS that uses the Sandholm algorithm and the
concept of a Generalized Vickrey Auction (GVA)~\cite{narahari05}.
Sandholm's algorithm is a method to determine the optimal allocation
of resources~\cite{sandholm02} in a CAS.  The concept of single-round
second-price sealed-bid auction is then used to determine the payment
made by the winners.  According to this, the payment made by a winner
is determined by the second-highest bid.  In order to achieve fairness
in such a CAS, we extend this existing payment scheme and take into
consideration the fair values of resources as perceived by the bidders
and the auctioneer in the system.  Based upon their estimate of fair
values, payments are made by the winners.  A detailed analysis is done
to highlight some important properties exhibited by this extension of
the payment scheme.

We start by classifying fairness and explain its different notions in
Section~\ref{sec_fairness}. It is followed by our study on CAS in
Section~\ref{sec_cas} and mathematical formulations are given that are
used to extend the payment scheme to achieve fairness in CAS.
Section~\ref{sec_analysis} gives a detail analysis of the scheme that
highlights the attractive properties in our payment scheme.  We
conclude with Section~\ref{sec_conclusion} which offers some
conclusions about our efforts, and some suggestions for further work
along these lines.

\section{Classification of Fairness} \label{sec_fairness}

To explain the different notions of fairness in various MAS, we
classify fairness as \emph{Basic Fairness} and \emph{Extended
  Fairness}.  This section defines the various perceptions about
measuring fairness in MAS.

In our analysis, we do not consider agent preferences as being apart
from their bids, i.e., if an agent has a higher preference for
something, it is considered to indicate the same by a higher bid, and
vice versa.  All goods are considered divisible.

Our algorithm given in Section~\ref{fairness_algo} creates an
allocation that is seen as having fairness (either basic or extended)
by all agents in the system.

\subsection{Basic Fairness}

In many MAS, there occurs a need of allocating the resources in an
equitable manner, i.e., each agent gets an equitable share of the
resources. This happens mainly when every agent holds similar
significance for the given set of resources and has a desire to
procure it. Thus, it becomes necessary to allocate the resources in an
equitable fashion, i.e., such that each agent believes that its share
is comparable to the share of other agents. Thus, none of the agents
hold preferences over the share of other agents. Hence, we say that
every agent believes that the set of resources is divided fairly among
all the agents. This concept of fairness is termed as \emph{basic
  fairness}.

\begin{definition}
When allocation is perceived to be fair in
comparison to the other agents i.e. share of all the agents is
comparable, \emph{basic fairness} is said to be achieved in resource
allocation.
\end{definition}

This kind of fairness is required in the applications whereby fairness
is the key issue rather than the individual satisfaction of the
self-interested agents. In such applications, it becomes necessary to
divide a resource set in an equitable fashion so that every agent
believes that it is receiving its fair share from the set of
resources. Hence, we see that every agent enjoys material equality and
this ensures basic fairness among them. In other words, the concept of
basic fairness also ensures egalitarian social welfare~\cite{yann05}
and envy-freeness~\cite{brams05}.

An example of such application that pertains to the equitable
allocation of resources is given by Lematre~\cite{mara-survey}. It
deals with the equitable distribution of Earth Observing Satellite
(EOS) Resources. EOS is co-funded and exploited by a number of agents
and its mission is to acquire images of specific areas on earth
surface, in response to observation demands from agents. However, due
to some exploitation constraints and due to large number of demands, a
set of demands, each of which could be satisfied individually, may not
be satisfiable in a single day. Thus, exploitation of EOS should
ensure that each agent gets an equitable share in the EOS resources,
i.e., the demands of each agent is given equal weight assuming that
agents have equal rights over the resource (we assume that they have
funded the satellite equally). Hence, we observe that basic fairness
is achieved as the demands of all agents are entertained by the
equitable distribution of EOS resources.

\subsection{Extended Fairness}

In every MAS, we observe that each agent intends to procure a
resource at a value that is perceived by it to be fair for the
procurement. In other words, every agent assigns a fair value to
each resource that determines its estimate of the value of the
resource in quantitative terms. The fair value attached to each
resource can be expressed in monetary terms in most MAS. Thus, an
agent intends to procure a resource by trading it with cash which is
equal to the fair value attached to the resource by the respective
agent. In such cases, each agent believes that it procures the
resource at a fair value and, hence, believes the allocation to be
fair.

However, it is important to mention that the fair value attached to
each resource by an agent does not necessarily reflect the utility
value of the resource to it.  An agent may hold a higher or lower
utility value for a resource irrespective of the fair value attached
to the resource by it. Thus, the fair value attached to a resource is
an estimate of the actual value of the resource in the system as
perceived by an agent in quantitative terms. It means that an agent is
always willing to trade a resource at its fair value.

The resource procurement in such MAS is perceived to be fair by every
agent. Resources are allocated to the agents based upon different
criteria of optimality in a system. However, it is assured that each
agent that procures a resource perceives the trade to be fair. The
other aspects of allocation like resources procured by other agents,
fair values attached to the resource by other agents, utility value of
the resource to other agents, etc., are not considered while an agent
trades a resource with its fair value. Thus, we see that the kind of
fairness that is achieved in such system is irrespective of other
agents and hence, we term it as \emph{extended fairness}.

\begin{definition}
  When allocation is perceived to be fair by an individual agent
  procuring a resource, and is irrespective of the measures attached
  by other agents, \emph{extended fairness} is said to be achieved in
  resource allocation.
\end{definition}

\noindent
An example of such a system can be explained through a scenario of job
allocations in a multi-national company. Consider a MAS that refers to
a company hiring situation, comprising an agent offering the job
positions (i.e., the owner's agent) and a number of self-interested
agents who contend for these jobs. The contending agents express their
estimate of the fair value through their curriculum vitae that is
submitted to the owner agent, i.e., each contending agent believes
that its curriculum vita fulfills the minimum requirements for the job
and that it is eligible for the job.  Hence, the agents define their
perception of the required qualifications for the job through their
curriculum vitae and believe it to be sufficient to qualify for the
job.  The owner agent selects the job-seeker agent that holds at least
minimum qualifications required for the job but holds the maximum
qualifications among all the contending agents. Thus, the job is
allocated to the agent whose curriculum vita matches this criterion.
Hence, the allocation is perceived to be fair by the winning agent and
by all other agents as it is allocated to the most deserving among all
the agents.  Hence, the job is allocated to the agent on the basis of
its curriculum vita, i.e., an agent acquires a job at its estimate of
the fair value of the qualifications required for the job.

Thus, we see the two broad classification of fairness that explains
different notions of fairness as perceived by the agents in different
MAS. To explain these notions of fairness mathematically, we shall
study a framework where fairness is a required property in resource
allocation. However, we also see that resource allocation deals with
another key issue of optimality in various MAS. Thus, the best example
of resource allocation framework where both optimality and fairness
are the key issues is Combinatorial Auctioning Systems (CAS).

\section{Fairness in Combinatorial Auctioning Systems
  (CAS)} \label{sec_cas}

Combinatorial Auctioning Systems are a kind of MAS which comprise an
auctioneer and a number of self-interested bidders.  The auctioneer
aims at allocating the available resources among the bidders who, in
turn, bid for sets of resources to procure them in order to satisfy
their needs.  The bidders aim at procuring the resources at minimum
value during the bidding process, while the auctioneer aims at
maximizing the revenue generated by the allocation of these resources.
Thus, CAS refers to a scenario where the bidders bid for the set of
resources and the auctioneer allocates the same to the highest-bidding
agent in order to maximize the revenue. Hence, we see that optimality
is one of the key issues in CAS.  The Sandholm algorithm is used here
to attain optimal allocation of resources.  It works by making an
allocation tree and carrying out some preprocessing steps like pruning
to make the steps faster without compromising the
optimality~\cite{narahari05, sandholm02}.

However, besides optimality, another key issue desired by some
auctioning systems is fairness.  To incorporate this significant
property in this resource allocation procedure, we propose an
algorithm which uses a metric to measure fairness for each agent and
determines the final payment made by the winning bidders.

The algorithm that we describe is based upon a CAS that uses the
Sandholm algorithm for achieving optimality, and an
incentive-compatible mechanism called Generalized Vickrey Auction
(GVA) as the pricing mechanism that determines the payments to be
given by the winning bidders.  The Generalized Vickrey Auction (GVA)
has a payoff structure that is designed in a manner such that each
winning agent gets a discount on its actual bid.  This discount is
called a Vickrey Discount, and is defined in~\cite{narahari05} as the
extent by which the total revenue to the seller is increased due to
the presence of that winning bidder, i.e., the marginal contribution
of the winning bidder to the total revenue.

We give mathematical formulations to show that both kinds of fairness
can be achieved in CAS.  We show that \emph{extended fairness} is
achieved in all cases except in case of a tie, in which case
\emph{basic fairness} is ensured.

\subsection{Mathematical Formulation}

\subsubsection{Terminology}

Let our CAS be a multi-agent system which is defined by the following
entities:

\begin{itemize}

\item[(i)]A set $\Phi$ comprising $m$ resources \textit{\(r_0,
    r_1,\ldots, r_{m-1}\)} for which the bids are raised.

\item[(ii)] A set $\xi$ comprising $n$ bidders
  \textit{\(b_0,b_1,\ldots, b_{n-1}\)}.  These are the agents among
  whom the resources are allocated.

\item[(iii)] An auctioneer, denoted by $\lambda$, is the initial owner
  of all the resources and invites bids in the auctions.

\end{itemize}

Let us consider a CAS that comprises three bidders
\textit{\(b_0,b_1,b_2\)}, an auctioneer denoted as $\lambda$, and
three resources \textit{\(r_0, r_1, r_2\)}.  Each bidder is privileged
to bid upon any combination of these resources.  We denote the
combinations or subsets of these resources as \textit{\{\(r_0\)\},
  \{\(r_1\)\}, \{\(r_2\)\}, \{\(r_0, r_1\)\}, \{\(r_0, r_2\)\},
  \{\(r_1, r_2\)\}, \{\(r_0, r_1, r_2\)\}}. We shall use the term
package to define a set that comprises the subsets of resources won by
a bidder. For example, a package for a bidder winning the subsets
\textit{\{\(r_0\)\}} and \textit{\{\(r_1\)\}} is defined as
\textit{\{\{\(r_0\)\},
  \{\(r_1\)\}\}}.

Assume that the auctioneer and each bidder has fair valuation for each
of the individual resource (say, in dollars) as shown in
Table~\ref{table1}.

\begin{definition}
  The fair valuation for an agent represents its estimate of the
  actual value of the resource.
\end{definition}

Thus, fair valuation by a bidder and an auctioneer for each resource
represents their estimate of the actual value of each resource.  Thus,
a bidder is willing to trade a resource at its fair value and also
believes that no loss is incurred by the seller in the trade.
Similarly, the auctioneer is willing to sell a resource at the fair
valuation described for it by him. Fair value for a combination of
resources can be calculated as the sum of the fair value for each of
the resources in that combination. The fair valuation for a resource
by a bidder does not refer to the utility measure of the resource for
the bidder. We shall use the term fair valuation and fair value
interchangeably.

\begin{table*}[!h]
\centering
\begin{tabular}[h]{c|c c c}
\hline
 &\(r_0\)&\(r_1\)&\(r_2\) \\
\hline \hline
\emph{Bidder \(b_0\)}&5&8&8 \\
\emph{Bidder \(b_1\)}&10&2&8 \\
\emph{Bidder \(b_2\)}&10&5&10 \\
\emph{Auctioneer, $\lambda$}&8&10&15\\
\end{tabular}
\caption{Fair valuations for each resource by all bidders}
\label{table1}
\end{table*}

From Table~\ref{table1}, we can see that the bidder \(b_0\) values
resource \(r_0\) for \$5, \(r_1\) for \$8 and \(r_2\) for \$10. This
means that bidder \(b_0\) is willing to trade resource \(r_0\) with
\$5, \(r_1\) with \$8 and \(r_2\) with \$8 and believes that no loss
is incurred by the auctioneer in this trade. The fair valuation for
the subset \{\(r_0, r_2\)\} for the bidder \(b_0\) is calculated as
the sum of his the fair values for \(r_0\) and \(r_2\) i.e. 5 + 8 =
\$13. Similarly, fair valuation for a package is the sum of the fair
valuation of the comprising sets i.e. for a package \{\(r_0\)\},
\{\(r_1, r_2\)\}\}, the
fair value is the sum of the fair values of \{\(r_0\)\} and \{\(r_1, r_2\)\}.

Let the bids raised by the bidders for the individual resource and
different combination of resources be as given in table2. It can be
seen that the bids raised by each of the bidder for different sets of
resources may or may not be equal to the fair valuation of the
respective set of resources. A bidder can put zero bids for the set of
resources it does not wish to procure.

\begin{table*}[!h]
\centering
\begin{tabular}[h]{c|c c c c c c c}
\hline
 &$r_0$&\(r_1\)&\(r_2\)&\{\(r_0,r_1\)\}&\{\(r_0,r_2\)\}&\{\(r_1,r_2\)\}&\{\(r_0,
r_1,r_2\)\} \\
\hline \hline
\emph{Bidder \(b_0\)}&0&10&5&0&20&15&50\\
\emph{Bidder \(b_1\)}&10&5&10&30&0&0&50 \\
\emph{Bidder \(b_2\)}&10&0&15&20&30&0&30 \\
\end{tabular}
\caption{Bids raised by the bidders for different combination of resources} \label{table2}
\end{table*}

It is assumed that the bidding language used in our system is $OR$
bids, i.e., a bidder can submit any number of bids and is willing to
obtain any number of atomic bids for a price equal to the sum of their
prices~\cite{nisan00, narahari05, sandholm02}.  Recall that the set of
all the bids won by a bidder is referred to as a package.

\begin{itemize}

\item[(a)] A set $D$ which is a subset of the set of natural numbers,
  i.e., \(D \subseteq \mathbb{N}\), describing the possible values (in
  dollars) given to resources by bidders.

\item[(b)] A \emph{fairness matrix}, $\Gamma_{i,[1 \times m]}$, for
  the bidder $b_i$, and $\Gamma_{\lambda,[1 \times m]}$ for the
  auctioneer, $\lambda$, is defined as :

  \(\Gamma_i = [\tau_{i,0}, \tau_{i,1}, \ldots, \tau_{i,m-1}]\),   for the bidder $b_i$.\\
  \(\Gamma_\lambda = [\tau_{\lambda,0}, \tau_{\lambda,1}, \ldots, \tau_{\lambda,m-1}]\), for the auctioneer, $\lambda$.\\

  where the function $\tau_i$ is defined by a bidder, $b_i$, for a
  resource, $r_j$ as:
\begin{displaymath}
        \tau_i ( r_j) = d,   d \in D
\end{displaymath}

This function represents a fair valuation of a resource, $r_j$, by a
bidder $b_i$.  From table 1, we have $\tau_0 \left( r_1\right)$ = 8,
$\tau_1 \left( r_1\right)$ = 2, etc. Thus, from table 1, we have the
following fairness matrices: $\Gamma_0$ = [5, 8, 8]; $\Gamma_1$ = [10,
2, 8]; $\Gamma_2$ = [10, 5, 10]; $\Gamma_\lambda$ = [ 8, 10, 15]

\item[(c)] A function $\Upsilon_{i,k}$, known as the \emph{pay
    function} by a bidder, $b_i$ is defined as:
\begin{displaymath}
\Upsilon_{i,k} \left( b_i, \Psi_k\right) = d
\end{displaymath}

where \(Psi_k = \{\mu_j | \mu_j \in \mathrm{set \ of \ resources \ won
  \ by \ bidder} b_i\}\), and $\Upsilon_{i,k}$ is the cost of the
package, $\Psi_k$, to the bidder $b_i$ as calculated from the GVA
payment scheme.

\end{itemize}

\subsubsection{Algorithm To Incorporate Extended Fairness In
  CAS} \label{fairness_algo}
        
\begin{itemize}

\item[(1)] Each bidder and the auctioneer define its fairness matrix
  before the start of bidding process. It is a sealed matrix and is
  unsealed at the end of bidding process.

\item[(2)] An allocation tree is constructed at the end of the bidding
  process to determine the optimum allocation and the winning
  bidders~\cite{sandholm02}.  Information about all the bidders in a
  tie is not discarded using some pre-defined criteria.

\item[(3)] Use GVA pricing mechanism to calculate the Vickrey
  discount~\cite{narahari05} and, hence, payments by the winning
  bidders for their corresponding packages, i.e., calculate
  $\Upsilon_{ij}$ for the package $\Psi_j$ won by the bidder $b_i$.

\item[(4)] Calculate the fair value of the package won by each bidder
  and denote it as $\Pi_{ij}$ for the bidder $b_i$ who wins the
  package $\Psi_j$.

\item[(5)] Also calculate the fair value of each package using the
  fairness matrix of the auctioneer and denote it as $\Pi_{\lambda j}$
  for a package $\Psi_j$.

\item[(6)] Compare the values of $\Pi_{\lambda j}$ and $\Upsilon_{ij}$
  and determine the final payment by the bidder depending upon the
  following conditions:

\end{itemize}

Case 1: $\Upsilon_{ij} > \Pi_{\lambda j}$ Bidder pays the amount
$\Upsilon_{ij}$ and the auctioneer gains profit equal to
($\Upsilon_{ij} - \Pi_{\lambda j}$) which is distributed among other
bidders who bid for the package $\Psi_j$. The profit is distributed in
a proportional manner, i.e., in the ratio of $(\Pi_{kj} - \Pi_{\lambda
  j}) / (\Pi_{\lambda j})$ for a bidder $b_k$ who also bid for
$\Psi_j$ but is not a winning bidder.

Case 2: $\Upsilon_{ij} = \Pi_{\lambda j}$ In this case, the bidder
pays the amount $\Upsilon_{ij}$ to the auctioneer.

Case 3: $\Upsilon_{ij} < \Pi_{\lambda j}$ Auctioneer suffers a loss of
amount ($\Pi_{\lambda j} - \Upsilon_{ij}$).  However, loss can be
recovered as per the following cases:

\begin{itemize}

\item[(i)] $\Pi_{ij} > \Pi_{\lambda j}$ Bidder's estimate of fair
  valuation is more than $\Upsilon_{ij}$.  Thus, bidder gives the
  final payment of $\Pi_{\lambda j}$ to the auctioneer.

\item[(ii)] $\Pi_{ij} = Pi_{\lambda j}$ Bidder's estimate of fair
  value is same as that of auctioneer's estimate and is greater than
  the value $\Upsilon_{ij}$.  Thus, bidder pays amount $\Pi_{ij}$ to
  the auctioneer.

\item[(iii)] $\Pi_{ij} < \Pi_{\lambda j}$

\begin{itemize}

\item[(a)] $\Pi_{ij} \le \Upsilon_{ij}$ : then bidder's final payment
  remains the same, i.e., $\Upsilon_{ij}$

\item[(b)] $\Pi_{ij} > \Upsilon_{ij}$ : then bidder's final payment is
  equal to $\Pi_{ij}$.

\end{itemize}

\end{itemize}

\subsubsection{Handling the cases of tie - Incorporating Basic
Fairness}

Unlike traditional algorithms, we do not discard the bids in the cases
of a tie on the basis of some pre-decided criterion.  We consider
these cases in our algorithm to provide \emph{basic fairness} to
the bidders.

In cases of a tie, we shall measure the utility value of the
resource to each bidder in the tie.

\begin{definition}
  The utility value of a resource to a bidder is defined as the
  quantified measure of satisfaction or happiness derived by the
  procurement of the resource.
\end{definition}

Mathematically, we define utility value for a resource set $\mu_j$ as:

\[\upsilon_i(\mu_j ) = \nu_i(\mu_j) - \Pi_{ij}\]

where \(\nu_i(\mu_j)\) is the bid value of the resource \(\mu_j\) and
\(\Pi_{ij}\) is the fair valuation for the resource set $\mu_j$ for
the bidder $b_i$.

The bidders maximize this utility value to quantify the importance and
their need for the resource to them.  Thus, the higher the utility
value, the greater is the need for the resource set.

In such a case, fairness can be imparted if the resource set $\mu_j$
is divided among all the bidders in a proportional manner, i.e., in
accordance to the utility value attached to the resource by each
bidder.

Let us consider the same example to explain the concept of basic
fairness in our system. From table 2, we observe that the optimum
allocation attained through allocation tree comprises the resource
set $\{r_0, r_1, r_2\}$ as it generates the maximum revenue of \$50.
However, we see that this bid is raised by the two bidders, $b_0$
and $b_1$.

Thus, we calculate the fair value of the resource set
$\mu_1 = \{r_0, r_1, r_2\}$ for the bidder $b_0$ and $b_1$, i.e.,
$\Pi_{01}$ = 5+8+8 = \$21 and $\Pi_{11}$ = 10+2+8 = \$20. Thus, the
utility value of the resource set $\mu_0$ for the bidder $b_0$ and
$b_1$ is as follows:
        
\begin{itemize}

\item[] for bidder $b_0$, $\upsilon_0(\mu_1 )$ = 50 - 21 = \$29, and
\item[] for bidder $b_1$, $\upsilon_1(\mu_1 )$ = 50 - 20 = \$30.

\end{itemize}

Hence, the resource set $\mu_1$ is divided among bidders, $b_0$ and
$b_1$, in the ratio of 29:30. In other words, bidder $b_0$ gets
49.15\% and bidder b1 gets 50.85\% of the resource set $\mu_1$.

The payment made by the bidders is also done in the similar
proportional manner. For example, the bidders, $b_0$ and $b_1$, make
their respective payments in the ratio of 29:30 to make up a total of
\$50 for the auctioneer, i.e., bidder $b_0$ pays \$24.65 and bidder
$b_1$ pays \$25.35 to the auctioneer for their respective shares.

Hence, we see that extended fairness as well as basic fairness are
achieved in CAS by using a fairness metric.  We take into account the
fair estimates of the auctioneer and the bidders for each resource to
ensure that fairness is achieved to auctioneer as well as the bidders.

We shall do a detailed analysis of the new mechanism in the following
section.

\section{Analysis} \label{sec_analysis}

A detailed analysis is done to highlight some important concepts used
and the significant properties exhibited by our CAS through our
payment mechanism.

\subsection{Fairness}

In MAS, every agent has its own metric to measure fairness with
regards to the allocation of resources. In CAS, we see that the
auctioneer and the bidders have their own estimate of the fairness
value attached to each resource. We introduced the concept of
fairness matrix to attain the knowledge of the fair value attached
to each resource by the auctioneer and each bidder. This matrix is
used as a metric to ensure that each allocation of resources is
perceived to be a fair allocation by the bidder as well as the
auctioneer.

Thus, we say that extended fairness is achieved when a bidder procures
a resource for an amount that is equal to its estimate of fair value
of that resource. In such a case, the bidder believes that the
resource was procured by it at a fair amount irrespective of other
bidders' estimate of fair value of that resource. Thus, the allocation
is believed to be extendedly fair as per the
estimates of the winning bidder.

We also see that basic fairness is achieved in our system when
there is more than one bidder who has raised equal bid for the same
set of resources. In such a case, we divide the set of resources among
all the bidders so as to ensure fairness to all the bidders in a
tie. However, this division of resources set is done in a proportional
manner. We intend to divide the resource such that the bidder holding
highest utility value to it should get the biggest share. To ensure
this, we calculate the utility value (i.e., $\upsilon_i(\mu_j ) =
\nu_i(\mu_j) - \Pi_{ij}$) of the set of resources to each bidder and
divide the set in the ratio of these values among the respective
bidders. Thus, we see that each bidder procures its basic share of
the set of resources in accordance to the basic importance
attached by the bidder to the set of resources.

Due to the achievement of fairness through our payment scheme, the
bidders are expected to show willingness to participate in the
auctions.

\subsection{Rationality}

We shall see that the fairness matrix is a metric for fair valuation
that forces the bidders and the auctioneer to behave rationally.  In
other words, they attain maximum profits if they describe their fair
matrix truthfully. Our system ensures certain behavioral traits of
auctioneer and the bidders through which this property of rationality
is achieved in our system. These behavioral traits are described in
the following:

\begin{proposition}
  The auctioneer does not state extremely high or low values in its
  fairness matrix as this does not generate higher revenue.
\end{proposition}

\begin{proof}
  If an auctioneer states very high values in its fairness matrix,
  then Case 3 follows most of the times.  From Case 3, we observe that
  the auctioneer receives a payment equal to $\Pi_{\lambda j}$ only if
  this value is comparable to that of $\Pi_{ij}$ for a bidder
  $b_i$. In other words, an auctioneer benefits only if its valuation
  is not irrationally higher than that of the bidder. On contrary, the
  auctioneer does not state very low values in its fairness matrix.
  For such circumstances, Case 1 follows, whereby it seems to be that
  the auctioneer gains profit and, hence, it is distributed among the
  bidders.\end{proof}

\begin{proposition}
  Bidders do not state extremely high or low values in the fairness
  matrix as it does not help them procure the resources at lower
  values.
\end{proposition}

\begin{proof}

  We see that the Case 3 deals with the fairness values of the bidder
  $b_i$. In case $\Upsilon_{ij} < \Pi_{\lambda j}$ and $\Pi_{\lambda
    j} < \Pi_{ij}$, the bidder pays the amount $\Pi_{\lambda
    j}$. Otherwise if $\Upsilon_{ij} \le \Pi_{ij} \le \Pi_{\lambda
    j}$, the bidder pays the amount equal to $\Pi_{ij}$. In both the
  cases, we see that the value to be paid is higher than the bid
  value. However, if the bidder is in a tie for a resource set, then
  its utility value falls negative if $\Upsilon_{ij} \le
  \Pi_{ij}$. Hence, the bidder does not get the profits which are
  distributed among other bidders in a tie.  Thus, a bidder undergoes
  a loss if the value of $\Pi_{ij}$ is very high. On contrary, the
  bidder does not state lower values in the fairness matrix. In this
  case, a loss is perceived by the bidder
  under Case 3, condition (iii), part (a).\end{proof}

\begin{proposition}
Bidders raise their bids truthfully.
\end{proposition}

\begin{proof}
  Bidders gain by bidding truthfully. On bidding truthfully, they can
  maximize the Vickrey Discount on their bids. Secondly, in the cases
  of tie, they can maximize the profit earned ($\upsilon_i (\mu_j) =
  \nu_i (\mu_j) - \Pi_{ij}$), i.e., for a given value of $\Pi_{ij}$,
  profit can be maximized by raising the bids truthfully.\end{proof}

\subsection{Incentive Compatibility}

The payment mechanism described in our system is incentive compatible
in certain cases. In the cases, when payment value for a package, as
calculated from the VCG mechanism, is greater than the fair valuation
of the auctioneer for the same package, then Case 1 follows, i.e., the
auctioneer gets an amount higher than its fair valuation for that
package. It means that the auctioneer gains the profit equal to
($\Upsilon_{ij} - \Pi_{\lambda j}$). This profit is distributed among
the bidders who bid for the same package in the proportional manner as
explained in Case 1.

Thus, it also forces the bidders to bid truthfully so as to gain
maximum benefits from the auctioning system.

\subsection{Efficiency}

The cases of a tie are handled in such a way so as to ensure basic
fairness.  In such a case, we divide the resource in proportion to its
utility value to a bidder. Thus, a resource is allocated in accordance
to the wishes of the consumers and, hence, the net benefit attained
through its use is maximized. In other words, we can say that our
system is allocatively efficient as the resources are allocated to the
bidders who value them most and can derive maximum benefits through
their use. Hence, we achieve allocative efficiency by handling the
cases of tie in an efficient manner.

\subsection{Optimality}
    
Optimality is a significant property that is desired in a CAS.  We
ensure this property by the use of Sandholm algorithm in our
system. It is used to obtain the optimum allocation of resources so as
to maximize the revenue generated for the auctioneer.  Thus, output
obtained is the most optimal output and there is no other allocation
that generates more revenues than the current allocation.

\section{Conclusion} \label{sec_conclusion}

Thus, we have shown that fairness is incorporated in CAS, whereby all
the agents receive their fair share if they behave rationally.
Extended fairness as well as basic fairness is attained through our
payment mechanism. Optimal allocation is obtained through the Sandholm
algorithm and the other significant properties like allocative
efficiency and incentive compatibility are also achieved.  This is an
improvement because in the existing world of multi-agent systems,
there do not seem to be many studies that attempt to incorporate
optimality as well as fairness.  The present paper addresses this lack
in a specific multi-agent system, namely, the CAS.

However, this work can be extended towards achieving a generalized
framework suitable for all, or at least many, multi-agent systems,
rather than just CAS.

The framework described can also be extended in several ways: one is
to de-centralize the suggested algorithm, to avoid use of a single
dedicated auctioneer.  Especially in distributed computing
environments, it would be best for there to be a method to implement
the suggested algorithm (or something close to it) without requiring
an agent to act as a dedicated auctioneer.

A second important extension would be to find applications for the
work.  Some applications that suggest themselves include distribution
of land (a matter of great concern for governments and people the
world over) in a fair manner.  In land auctions where a tie occurs, no
pre-defined or idiosyncratic method need be used to break the tie;
rather, the allocation can be done fairly in the manner suggested.

Fairness is also an important and pressing concern in the computing
sciences and information technology, particularly, in distributed
computing~\cite{lamport2000}.  It is therefore also of interest to see
how our method for achieving fairness could be applied in such
contexts.

\bibliographystyle{siam}
\bibliography{fairness}

\end{document}